\documentclass{llncs}
\usepackage{enumerate}
\usepackage{latexsym}
\usepackage{amsmath}
\usepackage{amssymb}
\usepackage[latin1]{inputenc}
\usepackage{fullpage}
\usepackage{color}
\renewcommand{\epsilon}{\varepsilon}
\newcommand{\bestalpha}{\frac{4}{11}}
\newcommand{\bestalphadec}{0.3636}

\usepackage{tikz}
\usetikzlibrary{shapes}
\usetikzlibrary{patterns}
\usetikzlibrary{arrows,positioning,decorations.pathmorphing,trees}
\usetikzlibrary{angles,quotes}
\usetikzlibrary{decorations.pathreplacing}

\newenvironment{restatetheorem}[1]
{\begingroup
  \def\thetheorem{\ref{#1}}
  \begin{theorem}}
{\end{theorem}
 \addtocounter{theorem}{-1}
 \endgroup}

\usepackage[noend]{algorithmic}
\usepackage{algorithm}
\floatname{algorithm}{Procedure}

\newtheorem{numclaim}{Claim}

\begin{document}

\title{The Fair Division of Hereditary Set Systems}

\author{Z. Li\inst{1} \and
A. Vetta\inst{2}}

\institute{Computer Science Department, \'Ecole Normale Supérieure, Paris.\email{zl@zli2.com} \footnote{The first author thanks 
McGill University for hosting him while conducting this research.} \and
Department of Mathematics and Statistics,
and School of Computer Science, McGill University. \email{adrian.vetta@mcgill.ca}}

\maketitle

\begin{abstract}
  We consider the fair division of indivisible items using the maximin shares measure. Recent work on the topic 
has focused on extending results beyond the class of additive valuation functions. In this spirit, we study the case where the 
items form an hereditary set system. We present a simple algorithm that allocates each agent a bundle of items whose value is 
at least $\bestalphadec$ times the maximin share of the agent. This improves upon the current best known guarantee of $0.2$ due to Ghodsi et al.
The analysis of the algorithm is almost tight; we present an instance where the algorithm provides a guarantee
of at most $0.3738$. We also show that the algorithm can be implemented in 
polynomial time given a valuation oracle for each agent.
\end{abstract}

\section{Introduction}
Consider the problem of allocating $m$~heterogenous goods amongst $n$~agents.
How can this be achieved in an equitable manner?
This is the classical problem of {\em fair division} in economics and 
political science~\cite{Ste48}. The issue that arises immediately is how to define ``fairness".
Two important concepts that have been widely studied are {\em proportionality} and {\em envy-freeness}.
An allocation of the items to the agents is {\em proportional} if, for every agent, the value that the agent has for the {\em grand bundle} (all of the items) is
at most $n$ times greater than the value it has for the bundle it receives. 
The allocation is {\em envy-free} if the value an agent has for the bundle it receives is at least as large as
the value it has for the bundle of any other agent; that is, no agent is willing to exchange its allocated bundle for the 
bundle of another agent.\footnote{Observe that if the agents have sub-additive valuation functions then envy-freeness implies proportionality.}

Fair division has been extensively studied in the case of divisible items, typically in the 
case of a single heterogeneous good, namely {\em cake-cutting}~\cite{BK96,RW98}.
More pertinent to this work is {\em fair resource allocation}, the case of distinct but homogeneous goods.
There, for divisible items, general equilibria can provide fair allocations in restricted settings. For example, 
assume the agents have linear valuation functions. If each agent is now given the same budget then equilibrium 
prices exist where all items are completely sold and each agent receives
a most desired bundle; this concept of {\em competitive equilibrium from equal incomes} is due to Varian~\cite{Var74}.

In practice, however, the fair division of indivisible items is more important than that of divisible items.
This can be seen from the plethora of real-world examples, including course registration in universities, shift scheduling, draft assignment in sport, client assignment to 
sales-people, airport slot assignments, divorce settlements, and estate division~\cite{Bud11,KPW18}.
But, at first glance, it is not clear if anything useful can be said regarding the fair division of indivisible goods. 
For instance, what is a fair way to allocate a single indivisible good between two agents?
An important concept used in understanding the case of indivisible goods was introduced by 
Budish~\cite{Bud11}, namely, maximin shares. The basic protocol is familiar to every child when cake cutting: 
``I cut, you choose". More generally, for $n$~agents and $m$~indivisible goods, 
one agent partitions the items into $n$~bundles but that agent then gets the last choice of bundle.
Intuitively, a risk averse agent seeks a partition that maximizes the value of its least desired bundle
in the partition. The minimum value of a bundle in the optimal partition value is called the {\em maximin share} for the agent.
Clearly, since the agents have different valuation functions, the optimal partitions and the corresponding maximin share
values may differ for each agent.
The first question that then arises is whether one can partition the items in such a way that {\bf every} agent receives a bundle
whose value is at least its maximin share. The answer is {\sc no}, even for additive valuation functions~\cite{KPW18}.
This negative result leads to the question of whether or not approximate solutions exist. Specifically,
is there a partition that gives every agent a bundle of value at least an $\alpha$-fraction of their maximin share?
In a groundbreaking work, for additive valuation functions, Kurokawa, Procaccia and Wang~\cite{KPW18} showed the existence of a partition with $\alpha=\frac23$;
polynomial time algorithms with the same guarantee were subsequently given in~\cite{AMN17} and~\cite{BK17}. A stronger guarantee 
of $\alpha=\frac34$
was very recently obtained by Ghodsi et al.~\cite{SGH18}. More general classes of valuation function have also been studied.
Barman and Krishnamurthy \cite{BK17} proved a bound of $\alpha=\frac{1}{10}$ for the class of submodular valuation functions.
This was improved to $\frac{1}{3}$ by Ghodsi et al.~\cite{SGH18}, who also proved guarantees of~$\frac{1}{5}$ for fractionally subadditive (XOS)
valuations and $\Omega(1/\log n)$ for subadditive valuations.
 
\subsection{Our Results.}
In this paper, we consider the fair division problem in an {\em hereditary set system} (or {\em downward-closed} set system).
A set system $H=(J, \mathcal{F})$ consists of a set $J$ of items and a family $\mathcal{F}$ of feasible (independent) subsets of $J$.
The set system satisfies the {\em hereditary property}~if:
$$S\in\mathcal{F} \text{ and } T\subset S \Longrightarrow  T \in\mathcal{F}$$
Hereditary set systems are ubiquitous in computer science and optimization. 
They arise naturally in the presence of packing or cost constraints, for example in scheduling problems and
manufacturing processes~\cite{Pin12,Mal14}. Furthermore, they are of fundamental theoretic importance; notable combinatorial and geometric objects that
satisfy the hereditary property include matroids, simplicial complexes, and minor closed graph families such as networks embeddable on a surface.

In an hereditary set system, each agent $i$ has a value $v_{i,j}\ge 0$ for each item $j$
but these values are additive only on feasible sets in the set system. A formal description
of this model is given in Section~\ref{sec:model} along with a proposed algorithm for allocating the items amongst the agents.
Our main result, given in Section~\ref{sec:lower}, is that this algorithm provides a guarantee of at least~$\bestalphadec$ for
the maximin shares problem in an hereditary set system. This improves on the current best known bound of~$0.2$.
In Section~\ref{sec:upper} we prove that our bound is almost tight by constructing an instance where the algorithm has a 
performance guarantee of at most~$0.3738$.
Consequently, our lower and upper bounds for the performance guarantee of the algorithm are within an amount~$0.0072$.
The basic implementation of the algorithm runs in exponential time. So in 
Section~\ref{sec:fast} we show how to implement the procedure in polynomial time. Specifically, given  
a valuation oracle for each agent, the algorithm makes at most a polynomial in $m$ number of queries to the oracles and 
performs a polynomial amount of computation given the responses of the oracles.

\section{The Hereditary Maximin Share Problem}\label{sec:model}

In this section, we describe the maximin share problem on an hereditary set system. We present
a fair division algorithm for the problem and provide a simple performance analysis of the procedure
(which we improve upon in the next section).

\subsection{The Fair-Division Model.}
We have a set $I$ of $n$ agents and collection $J$ of $m$ items.
The items belong to an hereditary set system $H=(J, \mathcal{F})$ and agents desire 
feasible (independent) sets in the set system. Specifically, each agent $i$ has an additive valuation function over 
 independent sets. That is, each agent $i\in I$ has a value $v_{i,j}\ge 0$ for each item $j\in J$
 and, for any independent set $S\in \mathcal{F}$, we have $v_i(S)=\sum_{j\in S} v_{i,j}$. The value the agent has for a 
 set $S\notin \mathcal{F}$ is simply the maximum value
 it has for any feasible subset of $S$; that is $v_i(S)=\max\limits_{T\in \mathcal{F}: T\subset S} \sum_{j\in T} v_{i,j}$.

Our aim is to fairly divide up the items amongst the agents. We measure the fairness of a division with respect to the
maximin share of each agent. To define this, let $\mathbb{P}$ be the set of all partitions of the items into $n$ sets. 
The value of the \emph{maximin share} for a agent $i$ is then
$${\tt MMS}(i) \ =\  \max_{\mathcal{P} \in \mathbb{P}} \, \min_{P \in \mathcal{P}} \, v_i(P) .$$
That is, the maximin share is a partition that maximizes the value of the least valuable bundle
in the partition. A partition $\mathcal{P}_i=\{P_i^1,P_i^2,\dots, P_i^n\} \in \mathbb{P}$ that attains this value is called a 
\emph{maximin partition for agent $i$} and the elements of $\mathcal{P}_i$ are called maximin parts.
Observe that the maximin partition may be different for each agent.

Our objective is to find a partition of the items $\{S_1, S_2, \dots, S_\ell\}$ where the bundle $S_i$ allocated to agent $i$ 
has value at least its maximin share. In general this is not possible, so instead we search for approximate solutions.
Specifically, we desire the maximum fraction $\alpha > 0$ and an allocation $\{S_1, S_2, \dots, S_\ell\}$ 
such that $v_i(S_i) \ge \alpha \cdot {\tt MMS}(i)$, for every agent $i$. We call this the \emph{hereditary maximin share problem}.

The hereditary maximin share problem has a constant factor approximation. This is because our valuation functions 
are {\em fractionally subadditive} (XOS).
That is, the valuation function can be defined as the maximum over a collection of additive set functions. 
To show this, for each agent $i$, we define an additive function $a_i^S$ over the items for each independent set $S\in \mathcal{F}$.
Specifically, let
\[
a_i^S(j) 
= \begin{cases} v_{i,j} &\mbox{if } j\in S\\ 
0 & \mbox{if } j\notin S
\end{cases} .
\]

It is then easy to verify, for any set $T$ (independent or not), that 
$$v_i(T)\ =\ \max_{S\in \mathcal{F}}\, \sum_{j\in T} a_i^S(j) .$$
Thus $v_i$ is indeed fractionally subadditive (XOS). Using this fact, it follows as a special case of the result of Ghodsi et al. \cite{SGH18}
that a performance guarantee of $0.2$ is obtainable.
%\begin{theorem}\cite{SGH18}
%There is an algorithm for the fair division problem in hereditary set systems that allocates 
%every agent a bundle with an approximation guarantee $\alpha~=~\frac15$. \qed
%\end{theorem}

We remark that the valuation functions for hereditary set systems are not submodular functions.\footnote{For example,
consider an hereditary set system $H=(J, \mathcal{F})$ with three items $J=\{a,b,c\}$ and let the maximal independent 
sets in $\mathcal{F}$ be $\{a\}$ and $\{b, c\}$.
Suppose agent $i$ has item values $v_{i,a}=3, v_{i,b}=2$ and $v_{i,c}=2$. Thus $v_i(\{a,c\})=3$ and $v_i(\{a,b,c\})=4$. 
Consequently, the marginal value of adding item $c$ to
the set $\{a,b\}$ is larger than the marginal value of adding item $c$ to
the set $\{a\}$. Thus the valuation function is not submodular.}
The aim of this paper is to improve upon the $\alpha=0.2$ performance guarantee.

\subsection{A Fair-Division Algorithm.}
To obtain a better performance guarantee we apply a simple and natural procedure. To begin, without loss of generality, we 
may assume there are no agents with a maximin share of value 0; if so, such an agent may be allocated no items.\footnote{In fact,
the remaining agents will then obtain a stronger guarantee of at least $\alpha$ times their maximin share value assuming 
$n-1$ agents.} Then, by scaling we may assume that the maximin share of every agent is exactly~$1$. 
Even stronger, we may assume that, for every agent $i$, there exists a maximin partition such that the agent 
has value exactly~$1$ 
for each part in the partition. To see this formally, let $\mathcal{P}_i$ be a maximin partition for agent~$i$. 
Now define a new valuation function $\hat{v}_i$ with the property that $\hat{v}_{i,j} = \frac{v_{i,j}}{v_i(P)}$ if item $j$ is in part $P\in \mathcal{P}_i$. Thus $\hat{v}_i(P)=1$, for each part $P\in \mathcal{P}_i$. Furthermore, we have $\hat{v}_{i,j}\le v_{i,j}$ since, by the prior normalization, $v_i(P)\ge 1$ for each part $P\in \mathcal{P}_i$.
Thus any allocation of value at least $\alpha$ with respect to $\hat{v}_i$ 
is a factor~$\alpha$ allocation with respect to the true valuation $v_i$ because
$v_i(S_i) \ge \hat{v}_i(S_i) \ge \alpha =\alpha\cdot {\tt MMS}(i)$.
%\footnote{Indeed, if player $i$ is allocated items $S_i$ then we may assume $S_i$ is feasible (or we can replace it by a 
%feasible set of same value contained in $S_i$). Now
%  $\alpha \le \hat{v}_i(S_i)
%  = \sum_{P \in \mathcal{P}_i} \hat{v}_i(P \cap S_i)
%  = \sum_{P \in \mathcal{P}_i} v_i(P \cap S_i) v_i(P)
%  \le \sum_{P \in \mathcal{P}_i} v_i(P \cap S_i) \max_{P' \in \mathcal{P}_i} v_i(P')
%  = \sum_{P \in \mathcal{P}_i} v_i(P \cap S_i) MMS(i)
%  = v_i(S_i) MMS(i)$.}
Finally, we may assume that each part $P\in \mathcal{P}_i$ is an independent set. If not, we may replace $P$ by a
feasible subset of the same value. Note that this implies that not every item need be in the ``partition" $\mathcal{P}_i$.

We are now ready to present our fair-division algorithm which begins with the normalization above. 
This normalization is not required in our polynomial time algorithm; see Section~\ref{sec:fast}.
Given a target value $\alpha$, we search for a minimum cardinality feasible set
of at least the targeted value for some agent. If such a set is found we allocate that set (bundle) to that agent
and then recurse on the remaining items and agents. This method is formalized  in Procedure \ref{fairdivalg}.
\begin{algorithm}
\caption{The Fair Division Algorithm}\label{fairdivalg}
\begin{algorithmic}[PERF]
\STATE {\tt Input:} A set $I$ of agents, a set $J$ of items, and a target value $\alpha$.
\FOR{$\tau=1$ to $m$}
	\WHILE{{there exists a set $S\subseteq J$ with $|S|=\tau$ {\em and} an $i\in I$ with $v_i(S)\ge \alpha$}}
	 \STATE Allocate bundle $S$ to agent $i$
	  \STATE Set $I \leftarrow I\setminus \{i\}$
	  \STATE Set $J \leftarrow J\setminus S$
	\ENDWHILE
\ENDFOR
\end{algorithmic}
\end{algorithm}

We use the following notation. Let $\{S_1, S_2, \dots, S_\ell\}$ be the bundles assigned by the procedure in order; note that $\ell\le n$. 
We view the procedure as working in phases. In Phase $\tau$ the procedure searches for bundles of cardinality $\tau$ that
provide utility at least $\alpha$ for some agent; note that $\tau\le m$.
We denote by $\mathcal{A}_{\tau}$ the collection of all items allocated during Phase $\tau$. 

\subsection{A Simple Analysis.}
To begin, we present a very simple analysis that shows this algorithm gives a factor
$\alpha=\frac13$ guarantee. In Section~\ref{sec:upper}, we will give a more intricate and nearly tight analysis.

\begin{theorem}\label{thm13}
The procedure allocates every agent a bundle of value at least $\alpha=~\frac{1}{3}$.
\end{theorem}
\begin{proof}
Clearly, if an agent is allocated a bundle by the procedure then it receives a bundle of value at least
$\alpha$. So it suffices to show that the procedure allocates every agent a bundle if it is run with
a target value $\alpha=\frac{1}{3}$. 
For a contradiction, suppose the procedure terminates after allocating bundles to $\ell < n$ agents.
Let $i$ be an agent that is not allocated a bundle. 

We may assume that $\mathcal{A}_1=\emptyset$. That is, $v_{i,j}<\alpha$ for every agent $i$ and every item $j$
and so no items are allocated in Phase 1.
The argument is standard \cite{KPW18}: if a set of cardinality one is allocated to an agent then this item intersects at most one of
the $n$ bundles in the maximin partition of any other agent. Thus $n-1$ of the bundles in the partition are untouched and each still have total value $1$.
Consequently, $n-1$ agents remain and they each have a partition of the items into $n-1$ bundles each with value $1$.
Thus, we recurse on this smaller problem. 

Therefore, we may assume the procedure only allocates items in Phases $\tau\ge 2$. Since the algorithm considers bundles in increasing size $\tau$, agents receive a minimal bundle with value the sum of value of its elements (by definition of valuation).
Now take any set~$S$ allocated to some agent $k$ in Phase $\tau$.
It must be the case
that $v_i(S)<\frac{\tau}{\tau-1}\cdot \alpha$. If not then there is a set $T\subset S$ with cardinality $\tau-1$
such that $v_i(T)\ge \frac{\tau-1}{\tau}\cdot \frac{\tau}{\tau-1}\cdot \alpha =\alpha$. But, by the hereditary property,
the bundle $T$ is an independent set so should then have been allocated to agent $i$ in Phase $\tau-1$.
Let bundle $S_k$ be the $k$th bundle allocated, where $1\le k\le \ell$.
Let $U$ be the set of items unallocated by the procedure. 
Then the total value of unallocated items in some bundle of the maximin partition $\mathcal{P}_i$ of agent $i$ is at least
\begin{eqnarray}\label{eq:easy}
\sum_{j\in U\cap \mathcal{P}_i} v_{i,j}
\ \ge\  n-\sum_{k=1}^\ell v_i(S_k)
\ \ge\  n-\sum_{k=1}^\ell \frac{|S_k|}{|S_k|-1}\cdot \alpha
\ \ge\  n-\sum_{k=1}^\ell 2\cdot \alpha 
\end{eqnarray}
Here the final inequality arises as $\mathcal{A}_1=\emptyset$ and so $|S_k|\ge 2$ for every allocated bundle.
Since $\alpha=\frac13$, we obtain from (\ref{eq:easy}) that
\begin{eqnarray*}
\sum_{j\in U\cap \mathcal{P}_i} v_{i,j}
\ \ge \ n - \frac{2}{3} \ell 
\ >\  n - \frac{2}{3} n 
\ =  \frac{n}{3} 
\end{eqnarray*}
Because the maximin partition $\mathcal{P}_i$ contains exactly $n$ parts, there is a part that contains
a set $S$ of unallocated items where $v_i(S)\ge  \frac{n}{3}\cdot \frac{1}{n} = \frac{1}{3}$. By the hereditary property, this 
contradicts the fact that the procedure terminated without allocating agent $i$ a bundle.
Thus every agent received a bundle of value at least $\frac{1}{3}$. \qed
\end{proof}

\section{An Improved Lower Bound}\label{sec:lower}
In this section, we provide a much more detailed analysis of the fair division algorithm
and prove it provides for an approximation guarantee of $\alpha=\bestalphadec$.
This analysis is almost tight; in Section~\ref{sec:upper} we present an example showing that
the performance of the procedure is not better than $\alpha=0.3738$.
\begin{theorem}\label{improvement}
  The procedure allocates every agent a bundle of value at least $\alpha=~\bestalpha$.
\end{theorem}

Before proving this theorem, we give some intuition behind the analysis.
The basic approach is the same as in Theorem \ref{thm13}. For an appropriately chosen target value $\alpha$
we run the procedure and assume for a contradiction that some agent~$i$ was not allocated a bundle.
We then consider the maximin shares partition $\mathcal{P}=\{P_1,P_2,\dots, P_n\}$ for agent~$i$ and 
show that some part $P\in \mathcal{P}$ contains items with total value at least $\alpha$ that
are unallocated at the end of the procedure. By the hereditary property, this will contradict the fact the procedure terminated 
without allocating a bundle to agent~$i$.

However, in order for this method to work for $\alpha=\bestalpha$, we refine the analysis in four key ways.
To motivate these refinements, imagine the bundle assignment is determined by an adversary.
The adversary wishes to assign bundles $\{S_\ell\}_{\ell\neq i}$ to the other agents, in accordance with Procedure~\ref{fairdivalg}, such 
that in every part of $\mathcal{P}$ items of total value at least $1-\alpha$ are allocated.
Assuming that $W=\sum_{\ell\neq i} \sum_{j\in S_\ell} v_{i,j}$ then, from the perspective of the adversary, 
the best outcome is that this weight be spread evenly over the $n$ parts of maximin shares partition 
$\mathcal{P}$ for agent~$i$. The proof of Theorem~\ref{thm13} shows that $W$ can be at most $\frac23 n$.
That is, $W\le (1-\alpha)\cdot |\mathcal{P}|$, where $\alpha=\frac13$. 
This is illustrated in Figure~\ref{fig:1/3alg} where there are $9$ agents and $35$ items. 
The aim of the adversary is to cover all the shaded 
blue area using the red items, where the height of item $j_{r}$ is $v_{i,j_{r}}$.

\begin{figure}[h]
\begin{minipage}{5cm}
%\null\vfill
  \begin{center}
\begin{tikzpicture}[scale=1]
  \fill[blue!20!white] (0,0) rectangle (4.5,3);
  \draw (0,0) rectangle (0.5,4.5); \node [below] at (0.25,0) {{\footnotesize $P_1$}};
    \draw (0.5,0) rectangle (1,4.5); \node [below] at (0.75,0) {{\footnotesize $P_2$}};
      \draw (1,0) rectangle (1.5,4.5);  \node [below] at (1.25,0) {{\footnotesize $P_3$}};
        \draw (1.5,0) rectangle (2,4.5);  \node [below] at (1.75,0) {{\footnotesize $P_4$}};
          \draw (2,0) rectangle (2.5,4.5);  \node [below] at (2.25,0) {{\footnotesize $P_5$}};
            \draw (2.5,0) rectangle (3,4.5);  \node [below] at (2.75,0) {{\footnotesize $P_6$}};
              \draw (3,0) rectangle (3.5,4.5);  \node [below] at (3.25,0) {{\footnotesize $P_7$}};
                \draw (3.5,0) rectangle (4,4.5);  \node [below] at (3.75,0) {{\footnotesize $P_8$}};
                 \draw (4,0) rectangle (4.5,4.5); \node [below] at (4.25,0) {{\footnotesize $P_9$}}; 
   \node [left] at (0,0) {{\bf $0$}};
     \node [left] at (0,4.5) {{\bf $1$}};    
       \node [left] at (0,3) {{\bf $\frac23=1-\alpha$}};                   
 \draw[dotted, ultra thick] (0,3) --(4.5,3);
 \end{tikzpicture}
\end{center}
\end{minipage}
\quad \quad \quad \quad
\begin{minipage}{5cm}
  \begin{center}
\begin{tikzpicture}[scale=1]
 \fill[blue!20!white] (0,0) rectangle (4.5,3);
  \draw (0,0) rectangle (0.5,4.5); \node [below] at (0.25,0) {{\footnotesize $P_1$}};
    \draw (0.5,0) rectangle (1,4.5); \node [below] at (0.75,0) {{\footnotesize $P_2$}};
      \draw (1,0) rectangle (1.5,4.5);  \node [below] at (1.25,0) {{\footnotesize $P_3$}};
        \draw (1.5,0) rectangle (2,4.5);  \node [below] at (1.75,0) {{\footnotesize $P_4$}};
          \draw (2,0) rectangle (2.5,4.5);  \node [below] at (2.25,0) {{\footnotesize $P_5$}};
            \draw (2.5,0) rectangle (3,4.5);  \node [below] at (2.75,0) {{\footnotesize $P_6$}};
              \draw (3,0) rectangle (3.5,4.5);  \node [below] at (3.25,0) {{\footnotesize $P_7$}};
                \draw (3.5,0) rectangle (4,4.5);  \node [below] at (3.75,0) {{\footnotesize $P_8$}};
                 \draw (4,0) rectangle (4.5,4.5); \node [below] at (4.25,0) {{\footnotesize $P_9$}}; 

   \node [left] at (0,0) {{\bf $0$}};
     \node [left] at (0,4.5) {{\bf $1$}};    
       \node [left] at (0,3) {{\bf $1-\alpha$}};                   

    \filldraw[draw=black, fill=red!20!white] (0,0) rectangle (0.5,1); 
    \node [below] at (0.25,.75) {{\footnotesize $j_1$}}; 
   \filldraw[draw=black, fill=red!20!white] (0,1) rectangle (0.5,1.5); 
    \node [below] at (0.25,1.5) {{\footnotesize $j_{19}$}}; 
      \filldraw[draw=black, fill=red!20!white] (0,1.5) rectangle (0.5,3); 
    \node [below] at (0.25,2.5) {{\footnotesize $j_{32}$}}; 

    \filldraw[draw=black, fill=red!20!white] (1,0) rectangle (1.5,.75); 
    \node [below] at (1.25,.5) {{\footnotesize $j_{13}$}}; 
   \filldraw[draw=black, fill=red!20!white] (1,.75) rectangle (1.5,1.75); 
    \node [below] at (1.25,1.4) {{\footnotesize $j_{28}$}}; 
      \filldraw[draw=black, fill=red!20!white] (1,1.75) rectangle (1.5,3); 
    \node [below] at (1.25,2.5) {{\footnotesize $j_{30}$}}; 
 
     \filldraw[draw=black, fill=red!20!white] (2,0) rectangle (2.5,.5); 
    \node [below] at (2.25,.5) {{\footnotesize $j_4$}}; 
   \filldraw[draw=black, fill=red!20!white] (2,.5) rectangle (2.5,1.75); 
    \node [below] at (2.25,1.25) {{\footnotesize $j_{11}$}}; 
      \filldraw[draw=black, fill=red!20!white] (2,1.75) rectangle (2.5,3); 
    \node [below] at (2.25,2.5) {{\footnotesize $j_{18}$}}; 
    
        \filldraw[draw=black, fill=red!20!white] (3,0) rectangle (3.5,.75); 
    \node [below] at (3.25,.75) {{\footnotesize $j_3$}}; 
   \filldraw[draw=black, fill=red!20!white] (3,.75) rectangle (3.5,1.75); 
    \node [below] at (3.25,1.5) {{\footnotesize $j_{10}$}}; 
      \filldraw[draw=black, fill=red!20!white] (3,1.75) rectangle (3.5,2.5); 
    \node [below] at (3.25,2.25) {{\footnotesize $j_{20}$}}; 
          \filldraw[draw=black, fill=red!20!white] (3,2.5) rectangle (3.5,3); 
    \node [below] at (3.25,3) {{\footnotesize $j_{27}$}};
    
        \filldraw[draw=black, fill=red!20!white] (4,0) rectangle (4.5,.5); 
    \node [below] at (4.25,.5) {{\footnotesize $j_2$}}; 
   \filldraw[draw=black, fill=red!20!white] (4,.5) rectangle (4.5,1.25); 
    \node [below] at (4.25,1.15) {{\footnotesize $j_{17}$}}; 
      \filldraw[draw=black, fill=red!20!white] (4,1.25) rectangle (4.5,1.75); 
    \node [below] at (4.25,1.75) {{\footnotesize $j_{21}$}}; 
         \filldraw[draw=black, fill=red!20!white] (4,1.75) rectangle (4.5,2.5); 
    \node [below] at (4.25,2.5) {{\footnotesize $j_{34}$}}; 
      \filldraw[draw=black, fill=red!20!white] (4,2.5) rectangle (4.5,3); 
    \node [below] at (4.25,3) {{\footnotesize $j_{35}$}};

     \filldraw[draw=black, fill=red!20!white] (0.5,0) rectangle (1,0.5); 
    \node [below] at (0.75,.5) {{\footnotesize $j_5$}}; 
   \filldraw[draw=black, fill=red!20!white] (0.5,0.5) rectangle (1,1.25); 
    \node [below] at (0.75,1.15) {{\footnotesize $j_9$}}; 
      \filldraw[draw=black, fill=red!20!white] (0.5,1.25) rectangle (1,2); 
    \node [below] at (0.75,1.85) {{\footnotesize $j_{16}$}}; 
        \filldraw[draw=black, fill=red!20!white] (0.5,2) rectangle (1,3); 
    \node [below] at (0.75,2.75) {{\footnotesize $j_{24}$}};
 
      \filldraw[draw=black, fill=red!20!white] (1.5,0) rectangle (2,.6); 
    \node [below] at (1.75,.6) {{\footnotesize $j_{12}$}}; 
   \filldraw[draw=black, fill=red!20!white] (1.5,.6) rectangle (2,1.2); 
    \node [below] at (1.75,1.2) {{\footnotesize $j_{14}$}}; 
      \filldraw[draw=black, fill=red!20!white] (1.5,1.2) rectangle (2,1.8); 
    \node [below] at (1.75,1.8) {{\footnotesize $j_{15}$}}; 
        \filldraw[draw=black, fill=red!20!white] (1.5,1.8) rectangle (2,2.4); 
    \node [below] at (1.75,2.4) {{\footnotesize $j_{29}$}};
         \filldraw[draw=black, fill=red!20!white] (1.5,2.4) rectangle (2,3); 
    \node [below] at (1.75,3) {{\footnotesize $j_{31}$}};

         \filldraw[draw=black, fill=red!20!white] (2.5,0) rectangle (3,1); 
    \node [below] at (2.75,.75) {{\footnotesize $j_8$}}; 
   \filldraw[draw=black, fill=red!20!white] (2.5,1) rectangle (3,1.5); 
    \node [below] at (2.75,1.5) {{\footnotesize $j_{22}$}}; 
      \filldraw[draw=black, fill=red!20!white] (2.5,1.5) rectangle (3,2); 
    \node [below] at (2.75,2) {{\footnotesize $j_{26}$}}; 
        \filldraw[draw=black, fill=red!20!white] (2.5,2) rectangle (3,3); 
    \node [below] at (2.75,3) {{\footnotesize $j_{33}$}};
    
         \filldraw[draw=black, fill=red!20!white] (3.5,0) rectangle (4,1); 
    \node [below] at (3.75,.75) {{\footnotesize $j_6$}}; 
   \filldraw[draw=black, fill=red!20!white] (3.5,1) rectangle (4,1.5); 
    \node [below] at (3.75,1.5) {{\footnotesize $j_7$}}; 
      \filldraw[draw=black, fill=red!20!white] (3.5,1.5) rectangle (4,2); 
    \node [below] at (3.75,2) {{\footnotesize $j_{23}$}}; 
        \filldraw[draw=black, fill=red!20!white] (3.5,2) rectangle (4,3); 
    \node [below] at (3.75,2.75) {{\footnotesize $j_{25}$}};

  \draw[dotted, ultra thick] (0,3) --(4.5,3); 
 \end{tikzpicture}
\end{center}
\end{minipage}
\caption{An adversarial view of the proof of Theorem~\ref{thm13}.}\label{fig:1/3alg}
\end{figure}

To improve the bound, we show that the adversary cannot spread the weight of the items allocated to the other agents
in an even manner across the partition $\mathcal{P}$.
To prove this, the first refinement in the analysis is, upon termination of the procedure, rather than considering the 
entire maximin shares partition $\mathcal{P}=\{P_1,P_2,\dots, P_n\}$
for agent~$i$, we focus on a restricted {\em sub-partition} $\hat{\mathcal{P}}$ of $\mathcal{P}$. 
To find $\hat{\mathcal{P}}$ we use combinatorial arguments on an auxiliary graph that is constructed with respect to the allocation decisions
made in Phase 2 of the procedure. With this sub-partition $\hat{\mathcal{P}}$ more specialized
accounting techniques can then be applied. We explain how to find the sub-partition $\hat{\mathcal{P}}$ in Section~\ref{sec:subpartition}.

With no additional refinements we could evaluate $W_{\hat{\mathcal{P}}}=\sum_{\ell\neq i} \sum_{j\in S_\ell\cap \hat{\mathcal{P}}} v_{i,j}$,
the total value of allocated items in parts of $\hat{\mathcal{P}}$. If $W_{\hat{\mathcal{P}}}\le (1-\alpha)\cdot |\hat{\mathcal{P}}|$ then we 
obtain a guarantee of $\alpha$. But we can do better if the adversary was unable to spread this weight evenly over $\hat{\mathcal{P}}$.
So this second improvement is to incorporate this into our analysis. In particular, if the adversary allocates items from part 
$P\in \hat{\mathcal{P}}$ worth greater than $1-\alpha$ to agent $i$ then some of this weight was wasted from the perspective of the adversary. To quantify this, let 
$$v_P= 1-\sum_{\tau\ge 1}\, \sum_{j\in P\cap \mathcal{A}_{\tau}} v_{i,j}$$
be the total value of unallocated items in part $P\in \hat{\mathcal{P}}$ upon termination of the algorithm.
Recall above that $\mathcal{A}_{\tau}$ is the set of items allocated while searching for bundles of cardinality $\tau$ 
in Phase~$\tau$.
We then denote by $s_P=\alpha-v_P$ the {\em superfluity} of part $P$; from the perspective of the adversary,
the damage caused to agent~$i$ by this weight $s_P$ is superfluous. We may assume that $s_P > 0$, otherwise 
the procedure would have allocated agent $i$ a bundle. Accounting for this superfluous damage will be the second key 
ingredient in the proof. Superfluity is illustrated in Figure~\ref{fig:11/30alg} and will be studied in detail in Section~\ref{sec:keys}.
  \begin{figure}[h]
\begin{minipage}{5cm}
%\null\vfill
  \begin{center}
\begin{tikzpicture}[scale=1]
  \fill[blue!20!white] (1.5,0) rectangle (4.5,2.85);
  \draw (0,0) rectangle (0.5,4.5); \node [below] at (0.25,0) {{\footnotesize $P_1$}};
    \draw (0.5,0) rectangle (1,4.5); \node [below] at (0.75,0) {{\footnotesize $P_2$}};
      \draw (1,0) rectangle (1.5,4.5);  \node [below] at (1.25,0) {{\footnotesize $P_3$}};
        \draw (1.5,0) rectangle (2,4.5);  \node [below] at (1.75,0) {{\footnotesize $P_4$}};
          \draw (2,0) rectangle (2.5,4.5);  \node [below] at (2.25,0) {{\footnotesize $P_5$}};
            \draw (2.5,0) rectangle (3,4.5);  \node [below] at (2.75,0) {{\footnotesize $P_6$}};
              \draw (3,0) rectangle (3.5,4.5);  \node [below] at (3.25,0) {{\footnotesize $P_7$}};
                \draw (3.5,0) rectangle (4,4.5);  \node [below] at (3.75,0) {{\footnotesize $P_8$}};
                 \draw (4,0) rectangle (4.5,4.5); \node [below] at (4.25,0) {{\footnotesize $P_9$}}; 
   \node [left] at (0,0) {{\bf $0$}};
     \node [left] at (0,4.5) {{\bf $1$}};    
       \node [left] at (0,2.85) {{\bf $\frac{19}{30}=1-\alpha$}};     
     
  %   \draw[decorate,decoration={brace,amplitude=0.75cm}]  (0,0) -- (1,0);  
\draw[decorate,decoration={brace}, rotate around={180:(3,-.5)}]  (1.5,-.5) -- (4.5,-.5);                
 \node [below] at (3,-.5) {$\hat{\mathcal{P}}$};   
  
 \draw[dotted, ultra thick] (0,2.85) --(4.5,2.85);
 \end{tikzpicture}
\end{center}
\end{minipage}
\quad \quad \quad \quad
\begin{minipage}{5cm}
  \begin{center}
\begin{tikzpicture}[scale=1]
 \fill[blue!20!white] (1.5,0) rectangle (4.5,2.85);
  \draw (0,0) rectangle (0.5,4.5); \node [below] at (0.25,0) {{\footnotesize $P_1$}};
    \draw (0.5,0) rectangle (1,4.5); \node [below] at (0.75,0) {{\footnotesize $P_2$}};
      \draw (1,0) rectangle (1.5,4.5);  \node [below] at (1.25,0) {{\footnotesize $P_3$}};
        \draw (1.5,0) rectangle (2,4.5);  \node [below] at (1.75,0) {{\footnotesize $P_4$}};
          \draw (2,0) rectangle (2.5,4.5);  \node [below] at (2.25,0) {{\footnotesize $P_5$}};
            \draw (2.5,0) rectangle (3,4.5);  \node [below] at (2.75,0) {{\footnotesize $P_6$}};
              \draw (3,0) rectangle (3.5,4.5);  \node [below] at (3.25,0) {{\footnotesize $P_7$}};
                \draw (3.5,0) rectangle (4,4.5);  \node [below] at (3.75,0) {{\footnotesize $P_8$}};
                 \draw (4,0) rectangle (4.5,4.5); \node [below] at (4.25,0) {{\footnotesize $P_9$}}; 
   \node [left] at (0,0) {{\bf $0$}};
     \node [left] at (0,4.5) {{\bf $1$}};    
       \node [left] at (0,2.85) {{\bf $1-\alpha$}};                   

   \draw[decorate,decoration={brace}, rotate around={180:(3,-.5)}]  (1.5,-.5) -- (4.5,-.5);                
 \node [below] at (3,-.5) {$\hat{\mathcal{P}}$};      
 
    \draw[decorate,decoration={brace}, rotate around={180:(4.6,3.4)}]  (4.6, 2.85) -- (4.6,3.95);                
% \node [below, rotate around={90:(4.6,3.35)}] at (4.6,3.35) {superfluous}; 
  \node [right, rotate=-90] at (5,4.6) {superfluous}; 
 
     \filldraw[draw=black, fill=red!20!white] (2,0) rectangle (2.5,.5); 
    \node [below] at (2.25,.5) {{\footnotesize $j_4$}}; 
   \filldraw[draw=black, fill=red!20!white] (2,.5) rectangle (2.5,1.75); 
    \node [below] at (2.25,1.25) {{\footnotesize $j_{11}$}}; 
      \filldraw[draw=black, fill=red!20!white] (2,1.75) rectangle (2.5,3.1); 
    \node [below] at (2.25,2.5) {{\footnotesize $j_{18}$}}; 
    
        \filldraw[draw=black, fill=red!20!white] (3,0) rectangle (3.5,1); 
    \node [below] at (3.25,.8) {{\footnotesize $j_3$}}; 
   \filldraw[draw=black, fill=red!20!white] (3,1) rectangle (3.5,1.75); 
    \node [below] at (3.25,1.5) {{\footnotesize $j_{10}$}}; 
      \filldraw[draw=black, fill=red!20!white] (3,1.75) rectangle (3.5,2.5); 
    \node [below] at (3.25,2.25) {{\footnotesize $j_{20}$}};

        \filldraw[draw=black, fill=red!20!white] (4,0) rectangle (4.5,.5); 
    \node [below] at (4.25,.5) {{\footnotesize $j_2$}}; 
   \filldraw[draw=black, fill=red!20!white] (4,.5) rectangle (4.5,1.25); 
    \node [below] at (4.25,1.15) {{\footnotesize $j_{17}$}}; 
      \filldraw[draw=black, fill=red!20!white] (4,1.25) rectangle (4.5,2.25); 
    \node [below] at (4.25,2) {{\footnotesize $j_{21}$}}; 
         \filldraw[draw=black, fill=red!20!white] (4,2.25) rectangle (4.5,2.75); 
    \node [below] at (4.25,2.75) {{\footnotesize $j_{34}$}}; 
      \filldraw[draw=black, fill=red!20!white] (4,2.75) rectangle (4.5,3.75); 
    \node [below] at (4.25,3.5) {{\footnotesize $j_{35}$}};

      \filldraw[draw=black, fill=red!20!white] (1.5,0) rectangle (2,1); 
    \node [below] at (1.75,.75) {{\footnotesize $j_{12}$}}; 
   \filldraw[draw=black, fill=red!20!white] (1.5,1) rectangle (2,1.5); 
    \node [below] at (1.75,1.5) {{\footnotesize $j_{14}$}}; 
      \filldraw[draw=black, fill=red!20!white] (1.5,1.5) rectangle (2,2.25); 
    \node [below] at (1.75,2.1) {{\footnotesize $j_{15}$}}; 
        \filldraw[draw=black, fill=red!20!white] (1.5,2.25) rectangle (2,3.25); 
    \node [below] at (1.75,2.8) {{\footnotesize $j_{29}$}};
         \filldraw[draw=black, fill=red!20!white] (1.5,3.25) rectangle (2,4); 
    \node [below] at (1.75,3.85) {{\footnotesize $j_{31}$}};

         \filldraw[draw=black, fill=red!20!white] (2.5,0) rectangle (3,0.75); 
    \node [below] at (2.75,.6) {{\footnotesize $j_8$}}; 
   \filldraw[draw=black, fill=red!20!white] (2.5,.75) rectangle (3,1.5); 
    \node [below] at (2.75,1.4) {{\footnotesize $j_{22}$}}; 
      \filldraw[draw=black, fill=red!20!white] (2.5,1.5) rectangle (3,2); 
    \node [below] at (2.75,2) {{\footnotesize $j_{26}$}}; 
        \filldraw[draw=black, fill=red!20!white] (2.5,2) rectangle (3,3.35); 
    \node [below] at (2.75,2.8) {{\footnotesize $j_{33}$}};
    
         \filldraw[draw=black, fill=red!20!white] (3.5,0) rectangle (4,.75); 
    \node [below] at (3.75,.6) {{\footnotesize $j_6$}}; 
   \filldraw[draw=black, fill=red!20!white] (3.5,.75) rectangle (4,1.5); 
    \node [below] at (3.75,1.3) {{\footnotesize $j_7$}}; 
      \filldraw[draw=black, fill=red!20!white] (3.5,1.5) rectangle (4,2.25); 
    \node [below] at (3.75,2.1) {{\footnotesize $j_{23}$}};

  \draw[dotted, ultra thick] (0,2.85) --(4.5,2.85);
 
 \end{tikzpicture}
\end{center}
\end{minipage}
\caption{Superfluous damage to the sub-partition $\hat{\mathcal{P}}$.}\label{fig:11/30alg}
\end{figure}

 The third idea is to exploit any laxity the procedure provides before the start of the third phase.
 Essentially, the {\em laxity} $l_P$ of a part $P\in \hat{\mathcal{P}}$
 is a measure of how much better the unallocated value of the part is after Phase 2 than a ``perceived" worst case. 
 Equivalently, the laxity is measure of how poorly the adversary allocated items in Phase 2 if its goal is to cause the
 agent the worst possible damage. The concept of laxity is formalized in Section~\ref{sec:keys}. 
 
 Finally, the fourth key idea is to {\em amortize} our accounting process. Rather than simply focus independently on items in each 
 part of the partition~$\hat{\mathcal{P}}$, we will also redistribute values within allocated bundles that cross multiple parts 
 in the partition. The technical details of this amortization process is also given in Section~\ref{sec:keys}.
 The proof of Theorem~\ref{improvement} then follows in Section~\ref{sec:proof}.
(We remark that our proof makes no mention of an adversary because the bounds we present concerning the values in each 
part hold in every possible case including, of course, the worst case.)

\subsection{Finding a Sub-Partition.}\label{sec:subpartition}
Now assume the procedure terminates after $\ell < n$ iterations leaving at least one agent $i$ 
who does not receive a bundle. Let the maximin partition for agent $i$ be $\mathcal{P}=\{P_1,P_2,\dots, P_n\}$.
Let $J^*\subseteq J$ be the set of items
allocated to agents during the procedure. We will show
\begin{equation}\label{eq: stop}
\max_{1\le k\le n}\, v_i(P_k\setminus J^*)\ \ge \ \alpha
\end{equation}
This will contradict the fact that the procedure terminated without allocating a bundle to agent $i$.
So let's prove that inequality (\ref{eq: stop}) holds. 
As in the proof of Theorem~\ref{thm13}, without loss of generality, we may assume no items 
were allocated in Phase~1; that is, $\mathcal{A}_1~=~\emptyset$.
Next we construct an {\em auxiliary graph} $\mathcal{G}$ based upon the allocation decisions made in Phase~2.
The graph contains $n$~vertices, one vertex for each part $P_k$ in the partition $\mathcal{P}$.
The graph contains an edge connecting the two (possibly equal) parts containing the two items of the bundle, for each bundle allocated in Phase~2.
Observe that  a vertex in $\mathcal{G}$ may have degree greater than one since the part it represents may 
contain multiple items. This also implies that $\mathcal{G}$ may contain edges that are self-loops; this happens
whenever a bundle allocated in Phase 2 consists of two items in the same part of the partition $\mathcal{P}$.

Let's further investigate the structure of $\mathcal{G}$. Let $X$ be a maximal set of vertices of $\mathcal{G}$ that induce at least $|X|$ edges. 
Note that such a set exists because $\emptyset$ is a feasible choice for $X$. On  the other hand, it must be the case that $X \neq V(\mathcal{G})$.
This follows as $\mathcal{G}$ contains exactly $n$ vertices but at most $\ell < n$ edges.
But this, in turn, implies that $X$ induces exactly $|X|$ edges. If $X$ induced more than $|X|$~edges then we could add to it any
other vertex in $V(\mathcal{G})\setminus X$ and still maintain the desired property.

Now consider the subgraph $\mathcal{G}\setminus X$. This subgraph is a forest $F$. If it contained a cycle $C$
then $X\cup V(C)$ would contradict the maximality of $X$. Furthermore, there are no edges between $X$ 
and $\mathcal{G}\setminus X$; otherwise the endpoint in $\mathcal{G}\setminus X$ of such an edge 
could have been added to $X$.

Let the forest $F$ contain $s$ components consisting of a single vertex -- observe that, 
by the above argument, these vertices are also 
singleton components of $\mathcal{G}$. Let $F$ contain $c$~non-trivial components, that is, trees with at least one edge.
Clearly, every non-trivial tree contains at least two leaves. Therefore, we may select a set $Y$ that consists of every vertex
in non-trivial trees in $F$ except for exactly two leaves in each non-trivial tree.
Finally, we set $Z=V(\mathcal{G})\setminus (X \cup Y)$. An illustration of the auxiliary graph $\mathcal{G}$ and the 
sets $X, Y$ and $Z$ is shown in Figure~\ref{fig:aux}.

\begin{figure}[h]
\begin{center}
\includegraphics[scale=1.0]{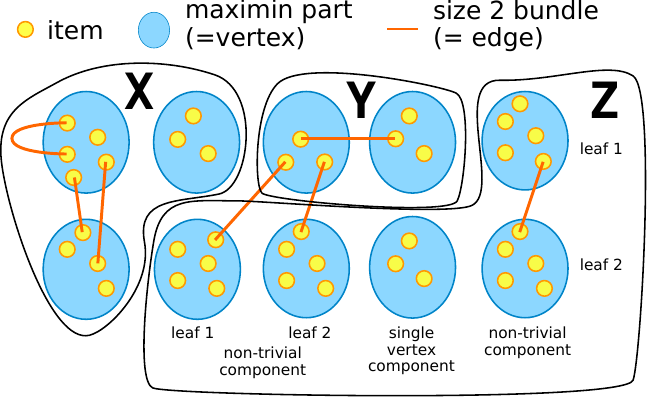}
\caption{\label{fig:cc} The Auxiliary Graph.}\label{fig:aux}
\end{center}
\end{figure}

The sub-partition $\hat{\mathcal{P}}$ of $\mathcal{P}$ consisting of the vertices in $Z=V(\mathcal{G})\setminus (X \cup Y)$
will be important to us.
Let's now present a couple of combinatorial equalities that will be useful later. 
The first is a claim that follows trivially by definition of $\hat{\mathcal{P}}$.
The second is a lemma quantifying how many agents are allocated bundles in Phase 2.
\begin{numclaim}\label{cl:Phat}
The number of parts in the sub-partition $\hat{\mathcal{P}}$ is $|Z|= 2c+s$. \qed
\end{numclaim}

\begin{lemma}\label{lem:phase2}
The number of agents allocated a bundle in Phase~2 is exactly $|X|+|Y|+c$.
\end{lemma}
\begin{proof}
Observe that, by construction, the number of agents allocated a bundle in Phase~2 is exactly $|E(\mathcal{G})|$.
So we must count the number of edges in the auxiliary graph $\mathcal{G}$.
We have seen that $E(\mathcal{G})=E(X)\cup E(F)$.
By the maximality of $X$, we have that $|E(X)|=|X|$. In addition, $|V(F)|=|Y|+2c+s$.
Thus, as $F$ consists of exactly $c+s$ trees, $|E(F)|=\left(|Y|+2c+s\right)-(c+s)= |Y|+c$.
Putting this together gives $|E(\mathcal{G})| = |X|+|Y|+c$, as desired.
\qed\end{proof}

We will focus our counting arguments on the  sub-partition $\hat{\mathcal{P}}$ in order to obtain a
contradiction to the fact that agent $i$ was not allocated a bundle.
Specifically, we will show that at least one of the vertices in $Z$ contains unallocated items that together provide value at least $\alpha$ to 
agent $i$.

\subsection{Laxity, Superfluity and Amortization.}\label{sec:keys}\ \\
Consider the allocated items in $\hat{\mathcal{P}}$. As $\mathcal{A}_1=\emptyset$, every item $j$ has $v_{i,j}< \alpha$. 
We now study the value (to agent $i$) of the items in $\hat{\mathcal{P}}$ allocated in Phase~2.
To do this, recall that the vertices of $\hat{\mathcal{P}}=Z$ are of two types, vertices of degree~$0$ in $\mathcal{G}$ 
(specifically, singleton vertices in $F$) and 
vertices of degree~$1$ (that is, leaf vertices in $F$). Vertices of degree $0$ contain no items that are allocated in Phase~2.
Vertices of degree~$1$ in $\mathcal{G}$ contain exactly one item that is allocated in Phase~2.
So given a part $P\in \hat{\mathcal{P}}$, we define the \emph{laxity} of $P$ to be $l_P=\alpha$ if $P$ corresponds to a singleton vertex in $F$.
Otherwise, if $P$ corresponds to a leaf vertex in $F$ we define $l_P=(1-v_{i,j^*(P)})-(1-\alpha)=\alpha-v_{i,j^*(P)}$, where
$j^*(P)$ is the unique item of $P$ that is allocated in Phase 2.
Observe that $l_P>0$ since $v_{i,j}<\alpha$ for every item $j$ and in particular for $j^*(P)$.
We can use the laxity to quantify the total value of items in $\hat{\mathcal{P}}$ allocated in Phase 2.
\begin{lemma}\label{lem:value-phase2}
  $$\sum_{P\in \hat{\mathcal{P}}} \sum_{j\in P\cap \mathcal{A}_2} v_{i,j}
  =2c\cdot \alpha -\sum_{P\in \hat{\mathcal{P}}, P\text{ is a leaf}} l_P$$.
\end{lemma}
\begin{proof}
We have 

$$
\sum_{P\in \hat{\mathcal{P}}} \sum_{j\in P\cap \mathcal{A}_2} v_{i,j}
\ =\  \sum_{P\in \hat{\mathcal{P}}, P\text{ is a leaf}}  \left(\alpha-l_P\right) 
\ =\ 2c\cdot \alpha -\sum_{P\in \hat{\mathcal{P}}, P\text{ is a leaf}} l_P
$$

Here the first equality holds since when $P$ corresponds to a singleton, the associated term in the sum is $0$, 
and when $P$ corresponds to a leaf, the associated term is $v_{i,j^*(P)} = \alpha - l_P$. The final equality follows 
as $\hat{\mathcal{P}}=Z$ contains exactly $2\cdot c$ leaves. 
\qed\end{proof}
Next we want to bound the value (to agent $i$) of items allocated in Phases 3 and beyond.
First, we count the number of bundles allocated in Phases 3 and beyond.
\begin{lemma}\label{lem:phase3+}
The number of agents allocated a bundle in Phases~3 and beyond is at most $c+s-1$.
\end{lemma}
\begin{proof}
By Lemma~\ref{lem:phase2}, the number of agents allocated a bundle in Phase~2 is exactly $|E(\mathcal{G})|=|X|+|Y|+c$.
Consequently, as agent $i$ is not allocated a bundle in the procedure, the number of agents allocated bundles in
Phases 3 and beyond is at most
\begin{eqnarray*}
(n-1)- |E(\mathcal{G})| &=& (n-1)-(|X|+|Y|+c)\\
&=& (|X|+|Y|+|\hat{\mathcal{P}}|-1)-(|X|+|Y|+c)\\
&=& |\hat{\mathcal{P}}|-1-c\\
&=& (2c+s)-1-c\\
&=& c+s-1
\end{eqnarray*}
Here the second equality arises because $n=|X|+|Y|+|Z|$ and $|Z|=|\hat{\mathcal{P}}|$;
the fourth equality follows from Claim~\ref{cl:Phat}.
\qed\end{proof}

Now we bound the value of bundles allocated in Phases 3 and beyond.
We will need two more definitions. First, recall, we defined the superfluity of a part $P$ as
$s_P= \alpha-\left(1-\sum_{\tau\ge 1}\, \sum_{j\in P\cap \mathcal{A}_{\tau}} v_{i,j}\right)$.
For the purpose of our analysis, we also define the \emph{excess} $e_P$ of a part $P\in \hat{\mathcal{P}}$ to be 
the sum of its superfluity and its laxity.
Therefore
\begin{eqnarray}\label{eq:excess}
e_P&=& s_P+l_P\nonumber \\
&=& \left(\alpha-1+\sum_{\tau\ge 1}\, \sum_{j\in P\cap \mathcal{A}_{\tau}} v_{i,j}\right)+ (\alpha-v_{i,j^*(P)})   \nonumber \\
&=& 2\alpha-1+\sum_{\tau\ge 3} \,\sum_{j\in P\cap \mathcal{A}_{\tau}} v_{i,j}
\end{eqnarray}
The final equality holds because $\mathcal{A}_1=\emptyset$ and $P\cap \mathcal{A}_2$ is either $\{j^*(P)\}$ or $\emptyset$.

As discussed, to bound the value of bundles allocated in Phases 3 and beyond, we will amortize our accounting process.
In these phases each allocated bundle has cardinality at least three. Take such a bundle, say $B=\{j_1,j_2,\dots, j_k\}$ 
allocated to some agent in phase $\tau$.
The hereditary property states that $B\setminus \{j\}$ is feasible, for every item $j\in B$. Thus, because $B$ is a minimum cardinality feasible
bundle of value at least $\alpha$ when it is allocated, it must be the case that
$v_i(B\setminus \{j\})<\alpha$.
Observe, this applies even for the least valuable item $\hat{j}\in B$
to agent $i$ in the bundle $B$. Furthermore, because $B$ is an independent set, 
we have that $(i)\ v_i(B\setminus \{\hat{j}\})=v_i(B)-v_{i,\hat{j}}$, and $(ii)\ v_{i,\hat{j}}$ is at most the average value of items in $B$.
Putting this all together gives
\begin{equation*}
v_i(B) \ <\  \alpha+v_{i,\hat{j}}\ \le\  \alpha+\frac{1}{k}\cdot v_i(B)\\
\end{equation*}
As $k\ge 3$, we obtain
\begin{equation}\label{eq:bundle-size}
v_i(B) \ <\  \frac{k}{k-1}\cdot   \alpha \ \le\  \frac{3}{2}\cdot   \alpha
\end{equation}
Hence, each bundle that is allocated in Phase 3 reduces the total value to agent $i$ of items in $\hat{\mathcal{P}}$ 
by at most $\frac{3}{2}\cdot \alpha$. For bundles of cardinality at least 4 (allocated in Phases 4 and beyond) this 
is at most $\frac{4}{3}\cdot \alpha$. The bound for bundles of size at least 4 is sufficient for out purposes, but the
bound for bundles of size 3 is not strong enough. So we now bound bundles allocated in Phase 3 more carefully.
To do so, we amortize the accounting process by defining, for any item $j$ allocated in a bundle $B$ by the mechanism,
$$a_j = \frac{v_i(B)}{|B|} .$$
That is, $a_j$ is the average value of items in $B$.
Since each allocated bundle is minimal and has cardinality~$3$, we have, for any item $j$, that
\begin{equation}\label{eq:bundle-size3}
a_j < \frac{1}{3}(\alpha + v_{i,j})
\end{equation}

We now further partition elements of $\mathcal{A}_{3}$ into two sets whose values we will bound independently in order to 
reduce the gaps in our accounting.
To this end, let $U$ consist of those items in $\mathcal{A}_{3}$ that belong to a 
part $P \in \hat{\mathcal{P}}$ which contains no other elements of $\mathcal{A}_{3}$. 
That is, $j\in U$ if item $j$ is the only item of $\mathcal{A}_{3}$ in some part $P \in \hat{\mathcal{P}}$.
Let $\bar{U} = \mathcal{A}_{3} \setminus U$; thus,  $\bar{U}$ consists of those items in $\mathcal{A}_{3}$ that belong to a 
part $P \in \hat{\mathcal{P}}$ which contains a least two items of $\mathcal{A}_{3}$.

\begin{numclaim}\label{cl:one-sixth}
If $\alpha\le \bestalpha$ then for any part $P$ of $\hat{\mathcal{P}}$ 
$$\sum_{j\in P\cap \bar{U}} a_j - \frac{1}{3}e_P \ \le\  \frac16\cdot |P\cap \bar{U}| .$$
\end{numclaim}
\begin{proof}
  Take any part $P \in \hat{\mathcal{P}}$. If $P$ contains no elements of $\bar{U}$ then the claim is trivially true as, by definition, $e_P \ge 0$. 
  Otherwise, by definition of $\bar{U}$, it must be the case that $P$ contains at least two elements of $\bar{U}$.
  Let $j_1$ be the most valuable and $j_2$ the second most valuable amongst items in $P \cap \bar{U}$. 
Assume $j_1$ was allocated in the bundle $B$ and $j_2$ was allocated in the bundle $B'$, where
possibly $B=B'$. By the minimality of $B$, we have $v_i(B)< \alpha +v_{i,j_1}$.
Since $j_1\in \bar{U}\subseteq \mathcal{A}_3$, we have that $|B|=3$ and so
\begin{equation}\label{eq:Bj}
a_{j_1} \ = \ \frac{v_i(B)}{|B|} \ < \  \frac{\alpha+ v_{i,j_1}}{|B|}  \  =  \  \frac{\alpha+ v_{i,j_1}}{3}
\end{equation}
By the same argument, $a_{j_2}   \le   \frac{\alpha+ v_{i, j_2}}{3}$. Hence,
\begin{eqnarray}
a_{j_1} + a_{j_2} - \frac13\cdot e_P 
 &\le& \frac{1}{3}(\alpha + v_{i,j_1}) + \frac{1}{3}(\alpha + v_{i,j_2}) - \frac{1}{3}e_P\nonumber \\
&=& \frac{2\alpha}{3} + \frac{1}{3}(v_{i,j_1} + v_{i,j_2} - e_P)\nonumber\\
&\le& \frac{2\alpha}{3} + \frac{1}{3}\left(1-2\alpha-\sum_{\tau\ge 3} \,\sum_{j\in P\setminus \{j_1, j_2\}\cap \mathcal{A}_{\tau}} v_{i,j} \right)
\end{eqnarray}
Here, the final inequality follows by (\ref{eq:excess}).
Thus we have
\begin{eqnarray}\label{eq:j1j2}
a_{j_1} + a_{j_2} - \frac13\cdot e_P
&\le& \frac{1}{3}- \frac{1}{3}\cdot \sum_{\tau\ge 3} \,\sum_{j\in P\setminus \{j_1, j_2\}\cap \mathcal{A}_{\tau}} v_{i,j}\nonumber \\
&\le& \frac{1}{3}- \frac{1}{3}\cdot \sum_{j\in P\setminus \{j_1, j_2\}\cap \mathcal{A}_3} v_{i,j}\nonumber\\
&\le& \frac{1}{3}- \frac{1}{3}\cdot \sum_{j\in P\setminus \{j_1, j_2\}\cap \bar{U}} v_{i,j}
\end{eqnarray}
Hence
\begin{eqnarray*}\label{eq:one-sixth}
  \sum_{j\in P\cap \bar{U}} a_j - \frac{1}{3} e_P
  &=& (a_{j_1} + a_{j_2} - \frac13\cdot e_P) +\sum_{j\in P\setminus \{j_1, j_2\}\cap \bar{U}} a_j \\
  &\le& \left(\frac{1}{3}- \frac{1}{3}\cdot \sum_{j\in P\setminus \{j_1, j_2\}\cap \bar{U}} v_{i,j} \right)+ \sum_{j\in P\setminus \{j_1, j_2\}\cap \bar{U}} a_j  \\
&\le& \left(\frac{1}{3}- \frac{1}{3}\cdot \sum_{j\in P\setminus \{j_1, j_2\}\cap \bar{U}} v_{i,j} \right)
+ \frac13\cdot\sum_{j\in P\setminus \{j_1, j_2\}\cap \bar{U}} (\alpha+v_{i,j}) 
\end{eqnarray*}
The above inequalities follow from (\ref{eq:j1j2}) and (\ref{eq:bundle-size3}), respectively.
Consequently, 
\begin{eqnarray}\label{eq:one-sixth2}
\sum_{j\in P\cap \bar{U}} a_j - \frac{1}{3} e_P  &\le&\frac13 +\frac13\cdot \sum_{j\in P\setminus \{j_1, j_2\}\cap \bar{U}} \alpha \nonumber\\
&\le&\frac13 +\frac13\cdot \sum_{j\in P\setminus \{j_1, j_2\}\cap \bar{U}} \alpha \nonumber\\
&=&\frac16 + \frac16 + \frac{\alpha}{3}\cdot (|P \cap \bar{U}|-2)\nonumber\\
&\le&\frac16 \cdot |P \cap \bar{U}| 
\end{eqnarray}
Here the final inequality follows as $\alpha=\bestalpha\le \frac12$. This completes the proof of the claim.
\qed\end{proof}

Applying Claim \ref{cl:one-sixth} to every part $P\in \hat{\mathcal{P}}$, we obtain that
\begin{eqnarray}\label{eq:Ubar}
  \sum_{j \in \bar{U}} a_j \ \le\  \frac{1}{6} \cdot |\bar{U}| + \sum_{P\in \hat{\mathcal{P}}} e_P
\end{eqnarray}
Moreover, because $a_j < \frac{1}{2}\cdot \alpha$ for all $j\in U\subseteq \mathcal{A}_3$, we also have
\begin{eqnarray}\label{eq:U}
  \sum_{j \in U} a_j \ <\  \frac{1}{2}\cdot \alpha \cdot |U|
\end{eqnarray}

\subsection{Proof of the Improved Bound.}\label{sec:proof}\ \\
We are now ready to prove the stated performance guarantee of the procedure.

\begin{restatetheorem}{improvement}
  The procedure allocates every agent a bundle of value at least $\alpha=\bestalpha$.
\end{restatetheorem}

\begin{proof}
Let $\beta_3$ be the number of bundles of size 3 allocated in Phase 3 and let $\beta_{4^+}$ be number of bundles of size at least 4
allocated in Phase 4 and beyond. Applying inequality~(\ref{eq:bundle-size}) for $k\ge 4$, we have
\begin{eqnarray}\label{eq:four-thirds}
\sum_{\tau\ge 4}\,  \sum_{j \in \mathcal{A}_{\tau}} v_{i,j} \ \le\  \frac{4}{3}\cdot \alpha \cdot \beta_{4^+}
\end{eqnarray}
This gives
\begin{eqnarray*}
\sum_{\tau\ge 3}\,  \sum_{j \in \mathcal{A}_{\tau}}   v_{i,j} &=& \sum_{j \in \mathcal{A}_{3}} v_{i,j} + \sum_{\tau\ge 4}\,  \sum_{j \in \mathcal{A}_{\tau}}  v_{i,j}\\
  &\le& \sum_{j \in \mathcal{A}_{3}} v_{i,j} + \frac{4 \alpha}{3}\cdot \beta_{4^+}\\
  &=& \sum_{j \in \mathcal{A}_{3}} a_j + \frac{4 \alpha}{3} \cdot \beta_{4^+}\\
  &=& \sum_{j \in U} a_j + \sum_{j \in \bar{U}} a_j + \frac{4 \alpha}{3} \cdot\beta_{4^+}
  \end{eqnarray*}
  where the last equality follows by definition of $U$ and $\bar{U}$.
  Applying (\ref{eq:U}) and (\ref{eq:Ubar}) then produces:
\begin{eqnarray*}  
 \sum_{\tau\ge 3}\,  \sum_{j \in \mathcal{A}_{\tau}}   v_{i,j} &\le& \frac{\alpha}{2}\cdot |U| +\left( \frac{1}{6} \cdot |\bar{U}| 
 + \frac{1}{3}\sum_{P\in \hat{\mathcal{P}}} e_P \right)+ \frac{4 \alpha}{3}\cdot\beta_{4^+}\\
  &=& \frac{\alpha}{2}\cdot |U| + \frac{1}{6} \cdot (|\mathcal{A}_{3}| - |U|) + \frac{4 \alpha}{3}\cdot \beta_{4^+} + \frac{1}{3} \sum_{P\in \hat{\mathcal{P}}} e_P \\
  &=& \left(\frac{\alpha}{2} - \frac{1}{6} \right)\cdot  |U| + \frac{1}{6} \cdot |\mathcal{A}_{3}| + \frac{4 \alpha}{3}\cdot \beta_{4^+}+ \frac{1}{3} \sum_{P\in \hat{\mathcal{P}}} e_P \\
  &=& \left(\frac{\alpha}{2} - \frac{1}{6} \right) \cdot |U| + \frac{1}{2} \cdot \beta_3 + \frac{4\alpha}{3}\cdot \beta_{4^+}+ \frac{1}{3} \sum_{P\in \hat{\mathcal{P}}} e_P \\
  &\le& \left(\frac{\alpha}{2} - \frac{1}{6} \right) \cdot |U| + \frac{1}{2} \cdot \left(\beta_3 + \beta_{4^+}\right) + \frac{1}{3} \sum_{P\in \hat{\mathcal{P}}} e_P 
\end{eqnarray*}
Here the third equality is due to the fact that $|\mathcal{A}_3|= 3\beta_3$; the final inequality follows as our target value is $\alpha=\bestalpha\le \frac38$.

Now, by definition, there is at most one element of $U$ for each of the $2c+s$ maximin parts in $\hat{\mathcal{P}}$.
So $|U| \le 2c+s$. Furthermore, by Lemma~\ref{lem:phase3+}, we have $\beta_3 + \beta_{4^+} \le c + s - 1$.
Therefore, it follows that
\begin{equation}\label{eq:tau3+}
\sum_{\tau\ge 3}\,  \sum_{j \in \mathcal{A}_{\tau}} v_{i,j} \ \le \ \left(\frac{\alpha}{2} - \frac{1}{6} \right) (2c+s)  + \frac{1}{2}(c+s)+ \frac{1}{3} \sum_{P\in \hat{\mathcal{P}}} e_P
\end{equation}
We are now ready to complete the proof. 
The total non-superflous value in $\hat{\mathcal{P}}$ at the end of the procedure is then
at most 
\begin{eqnarray*}
 \sum_{\tau\ge 1}\,  \sum_{j \in \mathcal{A}_{\tau}} v_{i,j} -\sum_{P\in \hat{\mathcal{P}}} s_P &=& \sum_{j \in \mathcal{A}_1} v_{i,j}
 	+ \sum_{j \in \mathcal{A}_2} v_{i,j}+ \sum_{\tau\ge 3}\,  \sum_{j \in \mathcal{A}_{\tau}} v_{i,j}-\sum_{P\in \hat{\mathcal{P}}} s_P\\
 &=& 0+ \left( 2\alpha c - \sum_{P\in \hat{\mathcal{P}}, P\text{ is a leaf}} l_P\right) + \sum_{\tau\ge 3}\,  \sum_{j \in \mathcal{A}_{\tau}} v_{i,j}-\sum_{P\in \hat{\mathcal{P}}} s_P\\
        &\le& 2\alpha c - \sum_{P\in \hat{\mathcal{P}}, P\text{ is a leaf}} l_P + \left( \left(\frac{\alpha}{2} - \frac{1}{6} \right) (2c+s) +\frac{1}{2} (c+s)+ \frac{1}{3} \sum_{P\in \hat{\mathcal{P}}} e_P\right) -\sum_{P\in \hat{\mathcal{P}}} s_P\\
  &\le& 2\alpha c - \sum_{P\in \hat{\mathcal{P}}, P\text{ is a leaf}} l_P + \left( \left(\frac{\alpha}{2} - \frac{1}{6} \right) (2c+s) +\frac{1}{2} (c+s)+ \frac{1}{3} \sum_{P\in \hat{\mathcal{P}}} e_P\right) -\sum_{P\in \hat{\mathcal{P}}} s_P\\
  \end{eqnarray*}
Here the second equality follows from the fact that $\mathcal{A}_1=\emptyset$ and by Lemma \ref{lem:value-phase2}. The inequality follows by (\ref{eq:tau3+}). We first bound the three terms containing laxity, superfluidity and excess. Since $\hat{\mathcal{P}}$ contains only singletons and leaves, we have

\begin{eqnarray*}
  &&- \sum_{P\in \hat{\mathcal{P}}, P\text{ is a leaf}} l_P + \frac{1}{3} \sum_{P\in \hat{\mathcal{P}}} e_P - \sum_{P\in \hat{\mathcal{P}}} s_P\\
  &\le& - \frac{1}{3}\sum_{P\in \hat{\mathcal{P}}, P\text{ is a leaf}} l_P + \frac{1}{3} \sum_{P\in \hat{\mathcal{P}}} e_P - \frac{1}{3}\sum_{P\in \hat{\mathcal{P}}} s_P\\
  &=& - \frac{1}{3}\sum_{P\in \hat{\mathcal{P}}, P\text{ is a leaf}} l_P \left(- \frac{1}{3}\sum_{P\in \hat{\mathcal{P}}, P\text{ is a singleton}} l_P
 + \frac{1}{3} \sum_{P\in \hat{\mathcal{P}}, P\text{ is a singleton}} l_P\right) + \frac{1}{3} \sum_{P\in \hat{\mathcal{P}}} e_P - \frac{1}{3}\sum_{P\in \hat{\mathcal{P}}} s_P\\
  &=& \frac{1}{3} \sum_{P\in \hat{\mathcal{P}}, P\text{ is a singleton}} l_P \\
  &=& \frac{1}{3} \sum_{P\in \hat{\mathcal{P}}, P\text{ is a singleton}} \alpha \\
  &=& \frac{1}{3} \alpha s
  \end{eqnarray*}

 The last two equalities follow from the definition of $l_P$ (for singletons) and the definition of $s$ respectively. Substituting and simplifying now gives
\begin{eqnarray*}  
 \sum_{\tau\ge 1}\,  \sum_{j \in \mathcal{A}_{\tau}} v_{i,j} -\sum_{P\in \hat{\mathcal{P}}} s_P   &\le& 2\alpha c  + \left(\frac{\alpha}{2} - \frac{1}{6} \right) (2c+s) +\frac{1}{2} (c+s) + \frac{1}{3}\alpha s\\
&=& 2c\cdot \left(\alpha + \frac{1}{2}\alpha - \frac{1}{6} +\frac{1}{4} \right) + s\cdot \left(\frac{1}{2}\alpha + \frac{1}{3}\alpha -\frac{1}{6} + \frac{1}{2} \right)\\
&=& 2c \cdot\left(\frac{1}{12}+\frac{3}{2}\alpha\right) + s\cdot \left(\frac{1}{3} + \frac{5}{6}\alpha \right)
\end{eqnarray*}
By Claim \ref{cl:Phat}, at the start of the procedure 
the total value of items to agent $i$ in the sub-partition $\hat{\mathcal{P}}$ is at least $(2c+s)\cdot 1$. Thus upon termination, the
total value of unallocated items (modulo superfluity) is at least
\begin{eqnarray*}
2c+s - \left( \sum_{\tau\ge 1}\,  \sum_{j \in \mathcal{A}_{\tau}} v_{i,j} - \sum_{P\in \hat{\mathcal{P}}} s_P\right)&\ge& (2c + s) 
- \left(2c \cdot\left(\frac{1}{12}+\frac{3}{2}\alpha\right) + s\cdot\left(\frac{1}{3} + \frac{5}{6}\alpha \right)\right)\\
&=& 2c \cdot\left(\frac{11}{12}-\frac{3}{2}\alpha\right) + s\cdot\left(\frac{2}{3} - \frac{5}{6}\alpha \right)\\
&\ge& (2c+s) \cdot\left(\frac{2}{3} - \frac{5}{6}\alpha \right) 
\end{eqnarray*}
Here the final inequality holds because $\frac{11}{12}-\frac{3}{2}\alpha \ge \frac{2}{3} - \frac{5}{6}\alpha$ for $\alpha=\bestalpha\le\frac38$.
Hence the average remaining value in each part of  $\hat{\mathcal{P}}$ at the end of the procedure is
at least $\frac{2}{3} - \frac{5}{6}\alpha$. This is at least $\alpha$ for $\alpha\le \bestalpha$.
Thus agent $i$ must receive a bundle of value at least $\alpha=\bestalpha=\bestalphadec$.
\qed\end{proof}

\section{An Upper Bound}\label{sec:upper}
We now show that the analysis in Section~\ref{sec:lower} is tight to within an additive amount of~$0.007$.
Specifically, we present an example that shows the procedure cannot guarantee a performance guarantee better than
$\frac{40}{107}=0.3738$.

\begin{theorem}
The performance guarantee of the procedure is at most $\frac{40}{107}=0.3738$.
\end{theorem}
\begin{proof}
Set $\alpha'=\frac{40}{107}$. We will construct an hereditary set system with the property that the procedure 
will fail to allocate every agent a bundle if we select a target value of 
$\alpha = \alpha'+\epsilon$, for any arbitrarily small constant $\epsilon>0$.

The set system contains six classes of items, denoted $\{A,B,C,D,E,F\}$.
The number of identical items in each class are shown in the second row of Table \ref{tab:hss}.
We assume the number $n$ of agents is a multiple of $2 \cdot 3 \cdot 5 \cdot 11 = 330$.
Moreover these agents are identical. The value each agent has for a single item of each class is shown in the
third row of the table. Finally, the feasible (independent) sets are defined by a capacity constraint for each
class of items, as shown in third row of Table \ref{tab:hss}. Specifically, a set $S$ is feasible if it contains at most
$2$ items of class $A$, at most $1$ item of class $B$, 
at most $2$ items of class $C$, at most $5$ items of class $D$, 
at most $11$ items of class $E$, and at most $40$ items of class~$F$. 

\begin{center}
\begin{table}[h]
\begin{tabular}{| c || c| c | c | c | c | c |}
\hline
{\tt Class}    & A      & B           & C & D & E & F\\ \hline
{\tt Quantity} & $n$ & $\frac13 n$ & $\frac23 n$ & $\frac23 n$ & $\frac13 n$ & $40n$ \\ \hline
{\tt Value} & $\alpha'=\frac{40}{107}$ &$\frac{13}{4}\alpha' - 1=\frac{23}{107}$  & $1-\frac{9}{4} \alpha'=\frac{17}{107}$ 
& $\frac14 \alpha'=\frac{10}{107}$ & $2 - \frac{21}{4} \alpha' = \frac{4}{107}$ & $\frac{1}{40}\alpha'=\frac{1}{107}$ \\ \hline 
{\tt Capacity} & 2 & 1 & 2 & 5 & 11 & 40 \\ \hline
\end{tabular}\\[0.2cm]
\caption{The Hereditary Set System}\label{tab:hss}
\end{table}
\end{center}

As feasibility is defined only by a capacity constraint on each item, it immediately follows that this set system
satisfies the hereditary property. (In fact, this set system is a {\em partition matroid} which prompts an interesting open problem
that we discuss in the conclusion.)

We claim that there is a partition $\mathcal{P}=\{P_1,P_2,\dots, P_n\}$ where each part has value exactly $1$.
Specifically, this maximin partition consists of parts of two types. There are $\frac23n$ parts of Type I: each part contains exactly one item from all the classes $A$, $C$, $D$ and 
forty items from class $F$. There are $\frac13n$ parts of Type II: each part
contains exactly one item from all the classes $A$, $B$, $E$ and forty items from class $F$.
Notice that in total we use $n$ items of class $A$, $\frac23n$ items from classes $C$ and $D$,
$\frac13n$ items from classes $B$ and $E$, and $40n$ items from classes $F$. Hence, this is a valid partition.
Moreover, each of these parts is an independent set as the capacity constraint for each class is clearly satisfied.
These parts all have value exactly one because
$$v(A,C,D,40F)=\frac{40}{107}+\frac{17}{107}+\frac{10}{107}+40\cdot\frac{1}{107}=\frac{107}{107}=1$$
and
$$v(A,B,E,40F)=\frac{40}{107}+\frac{23}{107}+\frac{4}{107}+40\cdot\frac{1}{107}=\frac{107}{107}=1 .$$

Now let's examine what the procedure does when faced with this instance.
It may allocate bundles as follows:
\begin{itemize}
\item Phase 2: $\frac12 n$ agents each receive $2$ items of class $A$.
\item Phase 3: $\frac13 n$ agents each receive $1$ item of class  $B$ and $2$ items of class  $C$.
\item Phase 5: $\frac{2}{15}n$ agents each receive $5$ items of class  $D$.
\item Phase 11: $\frac{1}{33}n$ agents each receive $11$ items of class  $E$.
\end{itemize}
To verify that this is a valid output of the procedure, it is necessary to check that all the bundles above were 
of minimum cardinality at the time they were allocated. By construction, there are no items with value at least $\alpha$.
Therefore, no bundles are allocated in Phase~1. Thus the bundles consisting of two items of class $A$ are indeed
minimum cardinality feasible bundles when they are allocated in Phase 2.
After Phase~2, there are no items of class $A$ remaining. Thus, because of the capacity one constraint on items of class $B$, 
it is easy to see that every feasible bundle of value at least $\alpha$ must contain at least three items. In particular, taking one item from class $B$ and
one item of class $C$ gives total value $\frac{40}{107}=\alpha-\epsilon$. Thus, in Phase~3 the bundles consisting of $1$ item of class  $B$ and $2$ items of 
class $C$ are of minimum cardinality at the time of allocation. After $\frac13 n$ such bundles have been allocated then are now no items of classes $B$ or $C$;
hence only items of the classes $D,E$ and $F$ remain. As the maximum value of any item in these classes is $\frac{10}{107}$, it is now
not possible to obtain any feasible bundles of cardinality $4$ with value at least $\alpha$. Thus the bundles allocated in Phase 5 are
indeed of minimum cardinality; this exhausts the supply of items of class $D$.
Now the maximum value of a remaining item is only $\frac{4}{107}$ so there are no feasible bundles of cardinality at most $10$.
Consequently, the bundles allocated in Phase~11 are also all of minimum cardinality; these bundles also exhaust the supply of items of class $E$.

Notice, at this point, we have allocated bundles to $\left(\frac12+\frac13+\frac{2}{15}+\frac{1}{33}\right)\cdot n = \frac{329}{330}\cdot n$ of the agents.
Therefore, we must now verify why no further bundles are allocated in Phase 12 and beyond.
After Phase 11 all the items from classes $A,B,C,D$ and $E$ have been allocated.
Thus, only the items of class $F$ are left. But the capacity constraint on class $F$ implies that no agent
can obtain value $\alpha$ using only items in $F$. Therefore the procedure terminates at this point
with  $\frac{1}{330}n$ of the agents receiving no bundle.
\qed\end{proof}

\section{A Polynomial Time Implementation}\label{sec:fast}
Procedure \ref{fairdivalg} is not a polynomial time algorithm. However, it can be modified to give a polynomial time implementation given
access to a valuation oracle for each agent. 
To do this there are two main problems. First, the use of a phase $\tau$ to search for bundles of cardinality $\tau$ is clearly 
exponential time if the procedure ends up searching for bundles of large cardinality. Second, the procedure, requires the
maximin partition or, more specifically, the maximin share value for each agent. 

Here we detail how to overcome both these problems using a polynomial amount of computation and
a polynomial number of valuation queries. Specifically, we present an implementation of Procedure~\ref{fairdivalg}
that runs in time polynomial in the number of 
items $m$. (We may assume the number of items $m$ 
is at least the number of agents, otherwise each maximin share value is $0$.)
Note that the 
hereditary set system that is part of the input may be very large compared to the number of items.\footnote{Indeed, for $m$ items, 
there are doubly exponentially many possible hereditary set systems. For example, let $\mathcal{J}_{m/2}$ be the set of bundles of 
cardinality exactly $m/2$. Then any subset $\mathcal{J}$ of $\mathcal{J}_{m/2}$ induces a distinct hereditary 
set system. There are $2^{{m \choose m/2}}$ such set systems and so we would need at least ${m \choose m/2}$ bits to represent 
such an hereditary set system.} Therefore, henceforth, we make the standard assumption that the valuation functions are given via a 
valuation oracle. Specifically, for each agent $i$, we possess an oracle which returns $v_i(S)$ when given a subset $S$ of the items as a query.
Thus, our aim is to prove that the modified algorithm makes only a polynomial number of calls to the valuation oracles
and preforms only a polynomial amount of additional computations.

\subsection{Overview of the Polytime Algorithm.}
Before diving into the formal proof, we discuss the intuition behind the algorithmic modifications.
Recall there are two main problems we must overcome. First, if $\tau$ is large then searching over 
bundles of cardinality $\tau$ in Phase $\tau$ is exponential time.
This problem is quite easy to deal with. We implement Phases $1,2$ and $3$ as before but then we group together 
Phases $4$ through to $m$. We do this by finding a minimal cardinality bundle $S$ with value at least $\alpha$ for some agent.
That is, $v_i(S)\ge \alpha$ for some agent $i$, and $v_\ell(S\setminus\{j\})<\alpha$ for any agent $\ell$ and any item $j\in S$.
We then allocate bundle $S$ to a remaining agent who values it the most. 
Finding such a minimal set can be done in polynomial time.
Moreover, this grouping of phases does not affect the performance guarantee as the
proof of Theorem~\ref{improvement} already implicitly groups together Phases $\{4,5, \dots,m\}$ in the analysis.

The second problem concerns the fact that the procedure uses the maximin share value ${\tt MMS}(i)$ for each agent $i$.
Specifically, this is used to scale the valuation functions so that the total value of every part in the maximin share partition
is exactly one. Note that Procedure \ref{fairdivalg} does not require the maximin share partitions of the agents.
Moreover, the analysis in the proof of Theorem~\ref{improvement} requires only that value of each part is at least one (or at least 
${\tt MMS}(i)$ before scaling) rather than exactly one. So if we could calculate the ${\tt MMS}(i)$ then we would be done.
Unfortunately, we cannot efficiently calculate ${\tt MMS}(i)$ so we must estimate it. 
For additive valuation functions there is a PTAS for calculating the maximin share value \cite{Woe97} but this 
does not extend to hereditary set systems; so an alternative approach is required.
The basic approach we take is as follows. Suppose we had guesses $\mu_i$ for ${\tt MMS}(i)$ for each $i$ and that 
all these guesses are known to be {\em overestimates}. Then imagine running Procedure \ref{fairdivalg} with 
valuations $v_i' = v_i / \mu_i$. Next, suppose we get lucky and each agent is 
allocated a bundle. Now, because the $\mu_i$ are overestimates, each agent $i$ receives a bundle of 
value (with respect to $v_i$) of at least $\alpha \cdot \mu_i \ge \alpha \cdot {\tt MMS}(i)$.

But what do we do if some agent $i$ is not allocated a bundle? In this case we decrease $\mu_i$ to obtain new 
estimates and repeat the procedure again -- in fact, we can decrease $\mu_i$ for all agents who receive no bundle. 
We decrease $\mu_i$ by multiplying it by a constant factor $1 - \delta$, for some very small constant $\delta$.
After doing so $\mu_i$ does not \emph{underestimate} ${\tt MMS}(i)$ by worse than a $1 - \delta$ factor. 
The key to the proof is then to show that the algorithm must allocate agent $i$ a bundle if we 
underestimate ${\tt MMS}(i)$. Hence the procedure will never further decrease $\mu_i$ once it drops below ${\tt MMS}(i)$.
It then remains to show that  we only need to repeat the procedure at most a polynomial number of times; this will follow
by obtaining suitable upper and lower bounds on each ${\tt MMS}(i)$. 

\subsection{A Formal Description of the Polytime Algorithm.}
We now formalize the intuitive description given above and prove the resultant procedure runs in time polynomial in the number of items $m$. Formally, our polynomial time algorithm is described in Procedure \ref{full-fairdivalg}.

\begin{algorithm}[h]
\caption{Fair Division Algorithm (Fast Implementation)}\label{full-fairdivalg}
\begin{algorithmic}[PERF]
  \STATE {\tt Input:} A set $I$ of agents, a set $J$ of items, valuation functions $v_i$ for each $i \in I$, a target value $\alpha$, and a search ratio $1-\delta$.
\FORALL{$i\in I$}
\STATE Set $\mu_i \leftarrow m\cdot v_{i}(n)$, where $v_{i}(n)$ is the value of the $n$th most valuable item for agent $i$.
\ENDFOR
\LOOP
  \STATE Call Procedure \ref{fast-fairdivalg} on $v_i' = v_i / \mu_i$ to obtain an allocation $A$ and a collection $U$ of agents with no allocated bundle.
  \IF{$U = \emptyset$} \RETURN $A$ and $\mu_i$, $\forall i \in I$ \ENDIF
  \STATE Set $\mu_i \leftarrow (1-\delta)\mu_i$, $\forall i \in U$
\ENDLOOP
\end{algorithmic}
\end{algorithm}

The algorithm calls a sub-procedure, Procedure \ref{fast-fairdivalg}, 
which groups together Phases~$4$~to~$n$. This grouping, in turn, relies upon Procedure \ref{minset} which 
given a target value $\alpha$, finds a 
minimal cardinality bundle of items of value at least $\alpha$ and returns this bundle together with the agent who values it the most.

\begin{algorithm}[h!]
\caption{Allocation from estimates}\label{fast-fairdivalg}
\begin{algorithmic}[PERF]
\STATE {\tt Input:} A set $I$ of agents, a set $J$ of items, valuation functions $v_i'$ for each $i \in I$ and a target value $\alpha$.
\FOR{$\ell=1$ to $3$}
	\WHILE{{there exists a set $S\subseteq J$ with $|S|=\ell$ {\em and} an $i\in I$ with $v_i'(S)\ge \alpha$}}
	 \STATE Allocate bundle $S$ to agent $i$
	  \STATE Set $I \leftarrow I\setminus \{i\}$
	  \STATE Set $J \leftarrow J\setminus S$
	\ENDWHILE
\ENDFOR
\WHILE{{there exists an agent $i\in I$ with $v_i'(J)\ge \alpha$}}
	 \STATE Call \textsc{MinimalSet}$(I, J, v', \alpha)$ to find a minimal set $S$ of value at least $\alpha$ for some agent $k$
    \STATE Allocate bundle $S$ to agent $k$
	  \STATE Set $I \leftarrow I\setminus \{k\}$
	  \STATE Set $J \leftarrow J\setminus S$
	  \ENDWHILE
\RETURN The allocation made throughout this procedure.
\end{algorithmic}
\end{algorithm}

\begin{algorithm}[h]
\caption{{\sc MinimalSet}}\label{minset}
\begin{algorithmic}[PERF]
  \STATE {\tt Input:} A set $I$ of agents, a set $J$ of items, valuation functions $v_i'$ for each $i \in I$ and a target value $\alpha$.
  \STATE {\tt Output:} A minimal set $S$ of value at least $\alpha$ for some agent $k$
  \STATE \(S \leftarrow J\)
  \LOOP
  \IF{there is an agent \(i \in I\) and element \(s\) of \(S\) such that
    \(v'_i(S \setminus \{s\}) \ge \alpha\)}
    \STATE \(S \leftarrow S \setminus \{s\}\).
  \ELSE
    \STATE \(i \leftarrow\) any agent in \(I\) maximizing \(v'_i(S)\)
     \RETURN \(S\) and \(i\).
  \ENDIF
  \ENDLOOP     
\end{algorithmic}
\end{algorithm}

To see the correctness of Procedure~\ref{minset}, first observe that it terminates.
This is because the number of remaining items decreases by at least one in each iteration.
Note that the if condition in Procedure~\ref{minset} is vacuously false once there are no items left. Furthermore, Procedure~\ref{minset} does return a minimal set $S$ of value at least $\alpha$ for some agent $i$. 
If not, suppose for a contradiction that $T \subset S$ has value $v'_{\ell}(T) > \alpha$ for some agent $\ell$. Then on the 
last iteration of Procedure~\ref{minset}, the condition of the if statement is made true with agent $\ell$ and any 
element $s \in T \setminus S$.

It remains to verify that every agent is assigned a bundle at the appropriate target $\alpha$.
This is shown by the following theorem whose proof is the same as Theorem~\ref{improvement}.
\begin{theorem}\label{polytime}
  Let $\mu_1, \ldots, \mu_n$ be positive constants.
  If we run Procedure \ref{fast-fairdivalg}, with valuations $v_i' = v_i / \mu_i$ and target $\alpha = \bestalpha$, then every 
  agent $i$ with $\mu_i \le {\tt MMS}(i)$ is allocated a bundle $S_i$ of value $v_i(S_i) \ge \alpha\cdot \mu_i$. \qed
\end{theorem}

The contrapositive of Theorem~\ref{polytime} is that any agent who does \emph{not} receive a bundle has $\mu_i > {\tt MMS}(i)$. 
Using this we may prove the main theorem of this section.
\begin{theorem}
Procedure \ref{full-fairdivalg} produces an allocation where every agent $i$ receives a bundle of value at least $(1-\delta)\cdot \alpha\cdot {\tt MMS}(i)$.
\end{theorem}
\begin{proof}
We prove this by induction on the number of iterations of the loop in Procedure~\ref{full-fairdivalg}. In fact, we prove the stronger 
statement that $\mu_i \ge (1-\delta) \cdot {\tt MMS}(i)$ for every agent $i$ throughout the procedure. Thus, in the last 
call to Procedure \ref{fast-fairdivalg}, each agent receives a bundle (as $U = \emptyset$) of value 
$\alpha \cdot \mu_i \ge (1-\delta)\cdot \alpha \cdot{\tt MMS}(i)$.
For the base case, before the first iteration, we have $v_i(J) \ge {\tt MMS}(i) > (1-\delta)\cdot {\tt MMS}(i)$ 
  for each agent $i$. Next suppose $\mu_i \ge (1-\delta) {\tt MMS}(i)$ holds in some iteration $\ell$; we now show it still holds in the 
  iteration $\ell +1$. There are two cases. First, if $i \not \in U$ then $\mu_i$ is unchanged in this iteration and so $\mu_i \ge (1-\delta)\cdot \alpha\cdot {\tt MMS}(i)$ is true in the next iteration.
Second, if $i \in U$ then, by the contrapositive of Theorem \ref{polytime}, we have that $\mu_i > {\tt MMS}(i)$. Then since $\mu_i$ is multiplied by a 
  factor $(1-\delta)$, in the next iteration we have $\mu_i \ge (1-\delta) \cdot {\tt MMS}(i)$ as required.
\qed\end{proof}

We remark that we can run Procedure \ref{full-fairdivalg} with $n$ identical copies of an agent to obtain a partition giving a $(1-\delta)\cdot \bestalpha$ 
approximation to the maximin share value of that agent.

\begin{theorem}
There is a $(1-\delta)\cdot\bestalpha$-approximation algorithm for the maximin share problem on hereditary set families for any $\delta > 0$.\qed
\end{theorem}

\subsection{The Running Time.}
Let's verify that this algorithm needs at most a polynomial number of valuation queries.

Procedure \ref{minset} requires at most $m$ iterations as we remove at least 
one item each iteration. In each iteration, it takes $|S| \le m$ queries to find whether there is an item removable from $S$
and we may attempt to do this for all $n$ agents. Thus, Procedure \ref{minset} runs in polynomial time.

Procedure \ref{fast-fairdivalg} first searches all bundles of size at most three using $O(n\cdot m^3)$ queries (by asking each agent for their value of 
each bundle of size at most three) and then runs at most $n$ iterations of Procedure \ref{minset}. 

The initialization step of Procedure \ref{full-fairdivalg}, for each of the $m$ agents, $n$ queries and $O(m)$ time is used to find the $n$th most 
valuable item.
Each iteration of Procedure \ref{full-fairdivalg} updates at most $m$ entries of $\mu$.
It remains to upper bound the number of iterations of Procedure \ref{full-fairdivalg} required. To compute this, we need 
upper and lower bounds on ${\tt MMS}(i)$.
\begin{numclaim}\label{cl:up-low}
Let $v_i(n)$ be the value of the $n$th most valuable item for agent $i$. 
Then $m\cdot v_i(n)\ge {\tt MMS}(i)\ge v_i(n)$.
\end{numclaim}
\begin{proof}
For the lower bound consider any partition $\mathcal{P}=\{P_1,P_2,\dots, P_n\}$ where $P_k$ contains
the $k$th most valuable item for agent $i$, for any $1\le k\le n$. Then clearly $v_i(P_k)\ge v_i(k)\ge v_i(n)$.
Thus, ${\tt MMS}(i)\ge v_i(n)$.

Next consider the upper bound. Take the maximin partition $\mathcal{P}$ for agent $i$. Then there is some part $P_k$ that does 
not contain any of the 
$n-1$ most valuable items for agent $i$. Since all the other items have value at most $v_i(n)$ we have that
$v_i(P_k)\le m\cdot v_i(n)$. Thus, ${\tt MMS}(i)\le m\cdot v_i(n)$.
\qed\end{proof}

The ratio between these upper and lower bounds for ${\tt MMS}(i)$ in Claim~\ref{cl:up-low} is $m$.
As our estimates for ${\tt MMS}(i)$ decrease in $(1-\delta)$ factors we have that the number of iterations required is
at most $n\cdot \log_{1/(1-\delta)} m=O(n\cdot \ln m)$.

Putting everything together we see that the total number of demand queries and the total computation time is polynomial in $m$.

\section{Conclusion}
We have presented a fair division algorithm for hereditary set systems which provides each agent with at least an
$\bestalpha$ fraction of its maximin share value. Several open problems remain. The first is to close the $0.007$ 
gap between the lower and upper bounds given for the performance of the procedure given in this paper.
The second is to design a new procedure with a better performance guarantee. Of course, this may be easier to do for 
sub-classes of hereditary set systems. One very important sub-class is that of matroids. A matroid is an hereditary set system that 
also satisfies the {\em augmentation property}:
given two independent sets $S$ and $T$ where $|S| > |T|$, there exists an element $s\in S$ such that $T \cup \{s\}$ 
is independent.
Because the almost-tight example presented in Section~\ref{sec:upper} is a (partition) matroid, the procedure presented in this paper does
not have a better performance guarantee for matroids. So a third open problem would be to design a fair division algorithm
that exploits the augmentation property to produce improved performance guarantees for matroids.\footnote{We remark that \cite{GM19} considers 
the maximin share problem in matroids but under a different model where, in addition to each agent being assigned a feasible set, 
the union of all these bundles must also be a feasible set.}

~\\
\noindent{{\sc Acknowledgements.}}
The authors thank Jugal Garg, Vasilis Gkatzelis and Richard Santiago for interesting discussions on fair division. 
We are also grateful to Dominik Peters for useful pointers to relevant literature.
We thank the anonymous reviewers for helpful suggestions.
Finally, we thank Zhou Yu for finding an error in an earlier version.

\nocite{*}
\bibliographystyle{plain}
\bibliography{fd-teac}

\end{document}